\documentclass[
aps,
prl,
reprint,
superscriptaddress,notitlepage,numbers]{revtex4-1}

\usepackage[caption=false]{subfig}
\usepackage{graphicx} 
\usepackage{epstopdf}
\usepackage{array}
\usepackage{verbatim}
\usepackage{amsmath,amsfonts,amssymb,amscd}
\usepackage{amsthm}
\usepackage{tabularx}
\usepackage{stmaryrd}
\usepackage{enumerate}
\usepackage[ruled]{algorithm2e}
\usepackage{enumitem}
\usepackage{wasysym}
\usepackage{braket}
\usepackage{color}
\usepackage{natbib}
\usepackage{newtxtext,newtxmath}

\definecolor{darkblue}{rgb}{0,0,0.5}
\usepackage{hyperref}
\hypersetup{
    bookmarksnumbered=true, 
    unicode=false, 
    pdfstartview={FitH}, 
    pdftitle={}, 
    pdfauthor={}, 
    pdfsubject={}, 
    pdfcreator={}, 
    pdfproducer={}, 
    pdfkeywords={}, 
    pdfnewwindow=true, 
    colorlinks=true, 
    linkcolor=black, 
    citecolor=darkblue, 
    filecolor=blue, 
    urlcolor=darkblue 
}
\usepackage{refcount}

\theoremstyle{plain}
\newtheorem{thm}{Theorem}

\newtheorem{lem}[thm]{Lemma}
\newtheorem{pro}[thm]{Proposition}

\theoremstyle{definition}

\newcommand{\eq}[1]{(\hyperref[eq:#1]{\ref*{eq:#1}})}

\renewcommand{\sec}[1]{\hyperref[sec:#1]{Section~\ref*{sec:#1}}}
\newcommand{\thrm}[1]{\hyperref[thrm:#1]{Theorem~\ref*{thrm:#1}}}
\newcommand{\lemm}[1]{\hyperref[lemm:#1]{Lemma~\ref*{lemm:#1}}}
\newcommand{\prop}[1]{\hyperref[prop:#1]{Proposition~\ref*{prop:#1}}}
\newcommand{\corr}[1]{\hyperref[corr:#1]{Corollary~\ref*{corr:#1}}}
\newcommand{\fig}[1]{\hyperref[fig:#1]{~\ref*{fig:#1}}}
\newcommand{\deff}[1]{\hyperref[deff:#1]{~\ref*{deff:#1}}}

\newcommand{\ff}{\mathcal{F}}

\newcommand{\mm}{\mathcal{M}}
\newcommand{\xx}{\mathcal{X}}
\newcommand{\yy}{\mathcal{Y}}

\DeclareMathAlphabet{\matheu}{U}{eus}{m}{n}

\DeclareMathOperator{\Tr}{Tr}

\renewcommand{\succ}{\text{\rm succ}}
\newcommand{\TT}{\mathcal{L}}

\newcommand{\ketbra}[2]{|{#1}\rangle\!\langle{#2}|}

\newcommand{\ba}{\begin{eqnarray}}
\newcommand{\ea}{\end{eqnarray}}
\newcommand{\bann}{\begin{eqnarray*}}
\newcommand{\eann}{\end{eqnarray*}}
\newcommand{\dm}[1]{\ketbra{#1}{#1}}
\newcommand{\bs}[1]{\boldsymbol{#1}}

\DeclareFontFamily{U}{mathc}{}
\DeclareFontShape{U}{mathc}{m}{it}%
{<->s*[1.03] mathc10}{}
\DeclareMathAlphabet{\mathscr}{U}{mathc}{m}{it}

\newcolumntype{L}[1]{>{\raggedright}p{#1}}
\newcolumntype{C}[1]{>{\centering}p{#1}}
\newcolumntype{R}[1]{>{\raggedleft}p{#1}}
\newcolumntype{D}{>{\centering\arraybackslash}X}

\begin{document}
\title{Operational Advantage of Quantum Resources in Subchannel Discrimination}

\author{Ryuji Takagi}
\email{rtakagi@mit.edu}
\affiliation{Center for Theoretical Physics and Department of Physics, Massachusetts Institute of Technology, Cambridge, Massachusetts 02139, USA}
\author{Bartosz Regula}
\email{bartosz.regula@gmail.com}
\affiliation{School of Mathematical Sciences and Centre for the Mathematics and Theoretical Physics of Quantum
Non-Equilibrium Systems, University of Nottingham, University Park, Nottingham NG7 2RD, United Kingdom}
\affiliation{School of Physical and Mathematical Sciences, Nanyang Technological University, 637371, Singapore}
\affiliation{Complexity Institute, Nanyang Technological University, 637335, Singapore}
\author{Kaifeng Bu}
\email{kfbu@fas.harvard.edu}
\affiliation{School of Mathematical Sciences, Zhejiang University, Hangzhou 310027, People's Republic of China}
\affiliation{Department of Physics, Harvard University, Cambridge, Massachusetts 02138, USA}
\author{Zi-Wen Liu}
\email{zwliu@mit.edu}
\affiliation{Perimeter Institute for Theoretical Physics, Waterloo, Ontario N2L 2Y5, Canada}
\affiliation{Center for Theoretical Physics and Department of Physics, Massachusetts Institute of Technology, Cambridge, Massachusetts 02139, USA}
\author{Gerardo Adesso}
\email{gerardo.adesso@nottingham.ac.uk}
\affiliation{School of Mathematical Sciences and Centre for the Mathematics and Theoretical Physics of Quantum
Non-Equilibrium Systems, University of Nottingham, University Park, Nottingham NG7 2RD, United Kingdom}

\begin{abstract}
One of the central problems in the study of quantum resource theories is to provide a given resource with an operational meaning, characterizing physical tasks in which the resource can give an explicit advantage over all resourceless states. We show that this can always be accomplished for all convex resource theories. We establish in particular that any resource state enables an advantage in a channel discrimination task, allowing for a strictly greater success probability than any state without the given resource. Furthermore, we find that the generalized robustness measure serves as an exact quantifier for the maximal advantage enabled by the given resource state in a class of subchannel discrimination problems, providing a universal operational interpretation to this fundamental resource quantifier. We also consider a wider range of subchannel discrimination tasks and show that the generalized robustness still serves as the operational advantage quantifier for several well-known theories such as entanglement, coherence, and magic.
\end{abstract}


\maketitle

{\textit{\textbf{Introduction.}}} ---
A rigorous understanding of quantum resources has been one of the ultimate goals in quantum information science. 
In addition to the apparent theoretical interest, it also has high relevance to burgeoning quantum information technologies such as quantum communication~\cite{kimble2008quantum,Duan2010}, quantum cryptography~\cite{bennett1984quantum,Gisin2002}, and quantum computation~\cite{shor1996fault,Gottensman1998}. 

Quantum resource theories~\cite{chitambar1806quantum} have recently attracted much attention as powerful tools which offer formal frameworks dealing with quantification and manipulation of intrinsic resources associated with quantum systems.
One could consider different theories depending on the relevant physical constraints, and indeed various resource theories have been proposed and analyzed, such as entanglement~\cite{plenio2007introduction,HOrodecki_review2009}, coherence~\cite{aberg2006quantifying,Baumgratz2014,Streltsov2017}, asymmetry~\cite{Gour2008,Marvian2016}, quantum thermodynamics~\cite{Brandao2013,Brandao_secondlaws2015}, non-Markovianity~\cite{wakakuwa2017operational}, magic~\cite{Veitch2014,howard_2017}, and non-Gaussianity~\cite{Genoni2008,Takagi2018,albarelli2018resource}.
Although these resource theories provide deeper insights into their specific physical settings, they do not tell us much about how to understand the individual properties and results in a unified fashion. 
In particular, despite the generality of the resource theoretical framework, only a small number of results reported in the literature are applicable to wide classes of general quantum resource theories~\cite{horodecki_2013,brandao_2015,Liu2017,Gour2017,regula_2018,anshu_2017,lami_2018}. 
In this work, we add a fundamental item to this list with regard to one of the central questions asked in the study of resource theories: the operational characterization of quantum states and the resources they possess.

An essential building block of a resource theory is the set of free states. 
It is the set of states that are considered ``easy to prepare'' in that theory, and any state outside of this set is called a resource state.  
A common and intuitive assumption is that the set of free states should be convex and closed. Convexity reflects a natural attribute in many physical settings, i.e. the fact that losing information about which free state was prepared, hence resulting in a probabilistic mixture of free states, should not by itself generate a resource.
Closedness, on the other hand, corresponds to the fact that the limit of a sequence of quantum states should accurately approximate the statistics of the states in the sequence for all physical experiments \cite{werner_1983}, which in particular implies that simply taking the limit should not create any resource.
To differentiate such theories from the few established resource theories which do not satisfy these constraints, and in particular, do not allow probabilistic mixing as a free operation \cite{Modi2012,Genoni2008,liu2017diagonal}, we will refer to any general theory obeying the conditions of closedness and convexity as a \textit{convex resource theory}.
 
In principle, one could define any set of free states and consider resource quantifiers defined with respect to this set \cite{horodecki_2013,regula_2018,chitambar1806quantum}. 
However, as the word ``resource'' suggests, it is desired that resource states should be useful for something; otherwise, the resource would lose physical significance and merely reduce to a mathematical concept. 
This question of operational characterization is always posed once the theory is proposed, and it is usually highly nontrivial. 
One of the ways to give an operational interpretation is to consider resource distillation~\cite{Bennett1996,Bravyi2005,Winter2016}. 
If a resource state can be distilled to a ``maximally resourceful'' state by free operations, that state can be associated with the tasks that utilize this unit state.    
However, whether there exists such an operational task is theory dependent, and furthermore some states cannot be distilled at all under some choices of free operations --- these are the bound resource states~\cite{Horodecki_bound1998,Campbell_bound2010,Veitch_bound2012,lami_2018,Takagi2018,Zhao_bound2018,lami_2018-1}. 
The latter fact makes the operational characterization even less clear for bound resources, even when the theory is physically well motivated. 

The question of operational significance of quantum resources has been addressed on a case-by-case basis. Of particular interest to us will be the task of channel and subchannel discrimination, a fundamental problem in quantum information theory \cite{kitaev_1997,childs_2000,acin_2001-1,watrous_2018}.
It has been demonstrated that, even without the aid of another state, every entangled state is useful in some channel discrimination task~\cite{Piani2009}, and the amount of entanglement of a state is directly related to its usefulness in channel discrimination~\cite{bae_2016,bae2018more}.
Analogous results have been shown also for steering, coherence and asymmetry~\cite{Piani2015,Napoli2016,Piani2016}, where it was not only shown that every resource state in these theories is useful in a particular subchannel discrimination task, but it was also found that the maximal advantage associated with a given state is exactly quantified by the measure known as the generalized robustness~\cite{vidal_1999,Steiner2003,harrow_2003}.
Although it would be natural to expect similar results to hold in more general cases, the arguments employed in the aforementioned works are specifically tailored to the above theories, and do not immediately generalize to encompass larger classes of resources.

Here, we show that {\it every} resource state in {\it any} convex theory is useful in a channel discrimination task, allowing for a strictly greater probability of success compared to discrimination using a free state, which gives an operational characterization to resource states in a theory-independent fashion.
As a result, we in particular provide an operational meaning to every bound resource state, including bound magic states~\cite{Campbell_bound2010,Veitch_bound2012} (see also Ref.\,\cite{Campbell2011}) as well as bound genuine non-Gaussian states~\cite{Takagi2018,nGnote}.  
We then find that the maximal advantage a resource state can provide in a class of subchannel discrimination problems is exactly quantified by the generalized robustness measure. 
The generalized robustness was first introduced as an entanglement monotone \cite{vidal_1999,Steiner2003,harrow_2003} and recently generalized to every finite-dimensional convex theory~\cite{regula_2018}.
Although the definition of this quantity is based primarily on geometric considerations, it is nevertheless known to admit operational interpretations in specific resource theories.
In the resource theory of coherence, as mentioned above, it characterizes the advantage a coherent state provides in subchannel discrimination tasks related to phase discrimination \cite{Napoli2016,Piani2016,Bu2017}, as well as quantifies the largest fidelity a state can achieve with the maximally coherent state in a single-shot transformation with free operations \cite{Bu2017,regula_2017}. Similarly, the generalized robustness of entanglement corresponds to the largest fidelity achievable with a maximally entangled state under free transformations \cite{regula_2018-1}.
The logarithmic version of this measure, known as the max-relative entropy \cite{datta_2009}, plays an essential role in the characterization of one-shot entanglement dilution \cite{brandao_2010,brandao_2011} and one-shot coherence dilution \cite{zhao_2018}, and quantifies the minimal rate of noise needed to catalytically erase the resource contained in a given state for a wider class of resource theories \cite{berta_2017-1,anshu_2017}.
However, a general operational meaning of the generalized robustness in all convex resource theories was not known.
Our result lifts the generalized robustness to an operationally meaningful measure in any convex resource theory, thus, generalizing and extending hitherto known results.
We finally consider relaxing the constraints placed on allowed measurements in the subchannel discrimination task and show that the maximal advantage is still quantified by the generalized robustness measure for some well-known theories such as entanglement, coherence, and magic.

{\textit{\textbf{All resource states are useful in a channel discrimination task.}}} ---
Let $L(\mathcal{X})$ be the set of linear operators acting on the Hilbert space $\mathcal{X}$, where the latter can be infinite dimensional, and let $\TT(\mathcal{X},\mathcal{Y}) = \{\Phi| \Phi:L(\mathcal{X}) \rightarrow L(\mathcal{Y})\}$ be the set of linear transformations that map the operators on the Hilbert space $\mathcal{X}$ to the operators on the Hilbert space $\mathcal{Y}$. Let $D(\xx)$ be the set of density operators acting on $\xx$, and $\ff(\xx)\subseteq D(\xx)$ be a closed and convex set. 
We say that if $\rho \in \ff(\xx)$, $\rho$ is a free state, and we call $\rho$ a resource state otherwise. 

Let $\{\Psi_i\}$ denote a finite set of subchannels (completely-positive trace-nonincreasing maps) that compose a completely positive trace-preserving (CPTP) map $\Lambda=\sum_i \Psi_i$ where $\Psi_i \in \TT(\xx,\yy)$. We consider a subchannel discrimination task where one is to decide which subchannel was applied to the input state $\rho\in D(\xx)$ by making a measurement on the output under the promise that only one of the subchannels in the set is realized.  
The goal of this task is to choose the best measurement strategy corresponding to a set of positive-operator valued measure (POVM) elements $\{M_i\}$, that maximizes the success probability $p_{\succ}(\{\Psi_i\},\{M_i\},\rho)=\sum_i\Tr(M_i\Psi_i(\rho))$. 
Note that channel discrimination, where one is to discriminate CPTP maps $\{\Lambda_i\}$ each of which is realized at the prior probability $p_i$, is a special case of subchannel discrimination where each subchannel is taken as $\Psi_i=p_i\Lambda_i$.

It was shown in Ref.\,\cite{Piani2009} that every entangled state is useful in a channel discrimination task. Translating this result to the framework of subchannel discrimination, the result says that for any entangled state $\rho$ there exists a channel discrimination task in which the quantity $\max_{\{M_i\}}p_{\succ}(\{\Psi_i\},\{M_i\},\rho)$ is strictly greater than $\max_{\{M_i\}}p_{\succ}(\{\Psi_i\},\{M_i\},\sigma)$ for any $\sigma \in \ff(\xx)$. We show the corresponding result for any choice of a convex and closed $\ff(\xx)$.

\begin{thm}
 Let $\rho \in D(\xx)$. Then, $\rho \notin \ff(\xx)$ if and only if there exist subchannels $\Psi_0,\Psi_1 \in \TT(\mathcal{X},\mathcal{Y})$ such that
 \ba
 \frac{\max_{\{M_i\}}p_{\succ}(\{\Psi_i\},\{M_i\},\rho)}{\sup_{\sigma\in\ff(\xx)}\max_{\{M_i\}}p_{\succ}(\{\Psi_i\},\{M_i\},\sigma)}>1.
 \nonumber
 \ea
\label{thrm:existence}
\end{thm}
\begin{proof}
The ``if`` direction of the Theorem is immediate.
For the other direction, note first by the Hahn-Banach separation theorem~\cite{hahn_banach} that, for any density operator $\rho\notin \ff(\xx)$, there exists a bounded self-adjoint operator $W\in L(\xx)$ such that $\forall \sigma\in \ff(\xx),\ \Tr(\sigma W) \geq 0$ and $\Tr(\rho W)<0$; conversely, if such an operator $W$ exists, then $\rho$ must be outside of the set $\ff(\xx)$.
We shall show that one can always construct two channels with equal prior probability such that $\rho$ gives an advantage in discriminating them. To this end, take another self-adjoint operator $X\in L(\xx)$ defined by $X=I-W/\|W\|_{\infty}\geq 0$ satisfying $\Tr(\rho X) > 1$ and $0\leq \Tr(\sigma X) \leq 1 \; \forall \sigma\in \ff(\xx)$, and consider the two maps $\Lambda_0, \Lambda_1 \in \TT(\mathcal{X}, \mathcal{Z})$ defined as
\bann
\Lambda_0(\eta) &\coloneqq& \left( \frac{\Tr(\eta)}{2} + \frac{\Tr(\eta X)}{2 \|X\|_\infty} \right) \dm{0} + \left( \frac{\Tr(\eta)}{2} - \frac{\Tr(\eta X)}{2\|X\|_\infty} \right) \dm{1}\\
\Lambda_1(\eta) &\coloneqq& \left( \frac{\Tr(\eta)}{2} - \frac{\Tr(\eta X)}{2 \|X\|_\infty} \right) \dm{0} + \left( \frac{\Tr(\eta)}{2} + \frac{\Tr(\eta X)}{2\|X\|_\infty} \right) \dm{1}
\eann
where $\mathcal{Z}$ is any Hilbert space of at least two dimensions containing the mutually orthogonal vectors $\{\ket{0}, \ket{1}\}$. It is straightforward to verify that $\Lambda_0, \Lambda_1$ are both completely positive trace-preserving maps, and thus valid quantum channels. Notice now that for any state $\rho$ we have $\|(\Lambda_0-\Lambda_1)[\rho]\|_{1} = 2\Tr(\rho X)/\|X\|_{\infty}$, which implies
\bann
 \begin{cases} \|(\Lambda_0-\Lambda_1)[\rho]\|_{1} \leq 2/\|X\|_{\infty}& \rho \in \ff(\xx),\\ \|(\Lambda_0-\Lambda_1)[\rho]\|_{1} > 2/\|X\|_{\infty} & \rho \notin \ff(\xx).\end{cases}
\eann
Consider now the task of discriminating the subchannel ensemble $\{\frac12 \Lambda_0, \frac12 \Lambda_1\}$, for which the maximal success probabiltiy is given by $\max_{\{M_i\}}p_{\succ}(\{\frac12 \Lambda_0, \frac12 \Lambda_1\},\{M_i\},\rho)=\frac{1}{2}(1+\|(\Lambda_0-\Lambda_1)[\rho]\|_{1}/2)$ by the Holevo--Helstrom theorem \cite{HOLEVO1973337,Helstrom}. The statement then follows immediately by noticing that for $\rho\notin\ff(\xx)$ and any $\sigma\in\ff(\xx)$ we have $\|(\Lambda_0-\Lambda_1)[\rho]\|_{1} > 2/\|X\|_{\infty} \geq \|(\Lambda_0-\Lambda_1)[\sigma]\|_{1}$.
\end{proof}

When $\xx$ is finite dimensional, the supremum in the statement of the Theorem is always attained.

We remark that the example subchannel discrimination task considered in the proof of the Theorem is in fact a binary channel discrimination problem, thus showing an advantage of any resource in the discrimination of quantum channels specifically.

This result is useful in the task of resource certification, where experimenters are to confirm that they truly possess a resource state. 
Indeed, the channel considered here has a direct connection to the witness operator that separates the resource state from the set of free states. This connection allows for another operational way of detecting a resource state in terms of channel discrimination, besides directly measuring the witness observable. Notably, due to the generality of the Theorem, this extends beyond the entanglement certification~\cite{guhne_2009,friis2018entanglement} to certifying other resources such as coherence, genuine non-Gaussianity, and magic.  

We further note that, by considering the assistance of ancillary systems, one could think of a more general setting where $\max_{\{M_i\}}p_{\succ}(\{\Psi_i\},\{M_i\},\rho)$ is compared to $\sup_{\sigma\in\ff(\xx\otimes\yy)}\max_{\{\tilde{M}_i\}}p_{\succ}(\{\Psi_i\otimes I\},\{\tilde{M}_i\},\sigma)$ where input free states are defined in the extended Hilbert space $\xx\otimes\yy$, and correspondingly $\{\tilde{M}_i\}$ is the set of POVMs acting on $\xx\otimes\yy$.
If $\ff(\xx\otimes\yy)$ allows for the entanglement between $\xx$ and $\yy$, the entanglement in the free states may help to distinguish the subchannels.
It is then not clear whether the same conclusion would still hold, as there might be a trade-off between the advantage provided by the resource in $\rho$ and the entanglement in $\sigma\in\ff(\xx\otimes\yy)$, which could be highly theory dependent. To consider explicitly the advantage provided by the resource itself, in this work we focus on the characterization of the resource in $\rho$ with respect to $\ff(\xx)$, but the above extension would certainly be interesting on its own and worth further study. 

{\textit{\textbf{Robustness as the advantage in subchannel discrimination.}}}--- 
Let $\xx$ be a Hilbert space with ${\rm dim}\xx = d < \infty$. Any closed convex set $\ff(\xx)\subseteq D(\xx)$ comes with the generalized robustness measure $R_{\ff(\xx)}:D(\xx) \rightarrow \mathbb{R}_+$ defined as 
\ba
R_{\ff(\xx)}(\rho)=\min_{\tau\in D(\xx)}\left\{s\left|\,\frac{\rho + s\tau}{1+s}\in \ff(\xx)\right.\right\},
\ea
It can also be obtained as the optimal value of the following convex optimization problem (see, e.g., Refs.\,\cite{brandao_2005,regula_2018}):
\ba
&{\text{\rm maximize}}& \ \ \Tr(\rho X) - 1 \label{eq:SDP_obj}\\
&{\text{\rm subject to}}&\ \ X\geq 0 \label{eq:SDP_cond1}\\
&& \ \ \Tr(\sigma X)\leq 1\ \forall \sigma \in \ff(\xx). \label{eq:SDP_cond2}
\ea
Note that the robustness can be infinite for some $\rho$ if $\ff(\xx)$ is composed only of rank-deficient states. Although our results can also be extended to such cases, we will hereafter assume that $\ff(\xx)$ contains at least one full-rank state for simplicity.

We shall find that the generalized robustness with respect to any choice of $\ff(\xx)$ allows for an operational interpretation: it serves as an exact quantifier for the advantage that a given state enables in a certain class of subchannel discrimination problems. Precisely, recall that the success probability in the discrimination of a set of subchannels $\{\Psi_i\}$ with the measurement strategy $\{M_i\}$ is given by $p_{\succ}(\{\Psi_i\},\{M_i\},\rho)=\sum_i\Tr(M_i\Psi_i(\rho))$. We will quantify the advantage that a quantum state $\rho$ provides over all free states $\ff(\xx)$ in the discrimination of $\{\Psi_i\}$ using the measurement strategy $\{M_i\}$ as the ratio of $p_{\succ}(\{\Psi_i\},\{M_i\},\rho)$ to the best success probability when using a free state, $\max_{\sigma \in \ff(\xx)} p_{\succ}(\{\Psi_i\},\{M_i\},\sigma)$.
The following result shows explicitly that, in any convex resource theory, the maximal such ratio optimized over all choices of sets of subchannels and measurement strategies is given precisely by the generalized robustness.
\begin{thm}
For any $\rho\in D(\xx)$,
\bann
 \max_{\{\Psi_i\},\{M_i\}} \frac{p_{\succ}(\{\Psi_i\},\{M_i\},\rho)}{\max_{\sigma \in \ff(\xx)} p_{\succ}(\{\Psi_i\},\{M_i\},\sigma)} = 1 + R_{\ff(\xx)}(\rho)
\eann
\label{thrm:rob_fixed}
\end{thm}
\begin{proof}
It can be easily shown that the left-hand side is less than or equal to the right-hand side as follows.
Recalling the definition of the generalized robustness, there exist $\tau\in D(\xx)$ and $\sigma\in \ff(\xx)$ such that $\rho = (1+R_{\ff(\xx)}(\rho))\sigma - R_{\ff(\xx)}(\rho)\tau$. Then, for any $\{\Psi_i\}$ and $\{M_i\}$,
\begin{equation}\begin{aligned}
 &p_{\succ}(\{\Psi_i\},\{M_i\},\rho) = \sum_i \Tr[M_i\Psi_i(\rho)] \\
 &\quad\leq (1+R_{\ff(\xx)}(\rho))\sum_i\Tr\left[M_i\Psi_i(\sigma)\right]\\
 &\quad\leq (1+R_{\ff(\xx)}(\rho)) \max_{\sigma \in \ff(\xx)} p_{\succ}(\{\Psi_i\},\{M_i\},\sigma).
\end{aligned}\end{equation}
Thus, it suffices to show that for any $\rho$, there exist $\{\Psi_i\}$ and $\{M_i\}$ such that $\frac{p_{\succ}(\{\Psi_i\},\{M_i\},\rho)}{\max_{\sigma\in\ff(\xx)} p_{\succ}(\{\Psi_i\},\{M_i\},\sigma)} \geq 1+ R_{\ff(\xx)}(\rho)$.
Let $X\in L(\xx)$ be an operator satisfying \eq{SDP_cond1} and \eq{SDP_cond2}. 
Let us write $X$ in its spectral decomposition as
$X=\sum_{i=1}^d x_i \dm{e_i}$
where $\{\ket{e_i}\}_{i=1}^d$ forms an orthonormal basis of $\xx$ and each $x_i \geq 0$. 
Consider now a set of unitaries $\{U_i\}_{i=1}^d$ such that 
$\sum_i U_i\dm{e_j}U_i^{\dagger} = I\ \forall j$ ---
the choice of such a set of unitaries is not unique, but there always exists one because we can, for instance, take $U_l \coloneqq \sum_{j=1}^d\ketbra{e_{j+l}}{e_j}$.
Now, consider the subchannels $\{\Psi_i\}$ defined by $\Psi_i(\cdot)=\frac{1}{d}U_i(\cdot)U_i^{\dagger}$ and measurement $\{M_i\}$ defined by $M_i=U_iXU_i^{\dagger}/\Tr(X)$. 
${M_i}$ is a valid POVM because $M_i\geq 0$ due to $X\geq 0$, and 
$\sum_iM_i = \frac{1}{\Tr(X)}\sum_{i} \sum_{j=1}^d x_j U_i\dm{e_j}U_i^{\dagger} 
 = \frac{1}{\Tr(X)}\sum_{j=1}^d x_j I = I.$ 
 This choice of subchannels and measurement gives $p_{\succ}(\{\Psi_i\},\{M_i\},\rho)=\Tr(\rho X)/\Tr(X)$ and
\bann
  \frac{p_{\succ}(\{\Psi_i\},\{M_i\},\rho)}{\max_{\sigma\in\ff(\xx)} p_{\succ}(\{\Psi_i\},\{M_i\},\sigma)} &=& \frac{\Tr(\rho X)}{\max_{\sigma\in\ff(\xx)}\Tr(\sigma X)}\\ &\geq& \Tr(\rho X).
\eann
The last inequality is due to \eq{SDP_cond2}.
The optimal $X$ satisfying \eq{SDP_obj}, \eq{SDP_cond1}, \eq{SDP_cond2} realizes $\Tr(\rho X)=1+R_{\ff(\xx)}(\rho)$, which concludes the proof.
\end{proof}

The generality of the result allows one to apply this to a variety of settings, and extends the operational connection between subchannel discrimination and resource witnesses to the so-called quantitative witnesses \cite{brandao_2005,eisert_2007,regula_2018}.
To exemplify the applicability of the Theorem, in the Supplemental Material we relate the result to an explicit physical problem of detecting the noise introduced by the application of a non-Clifford gate, of practical relevance for fault-tolerant quantum computation \cite{sm}.

{\textit{\textbf{Relaxation of measurement constraints.}}} ---
The result of \thrm{rob_fixed} gives an operational meaning to the generalized robustness in a very general fashion. 
However, one may also be interested in less restrictive settings of subchannel discrimination, where the measurement strategies for $\rho$ and for any free state $\sigma$ can be chosen independently.

Let us first consider the most general situation where, for each state, the experimenters can choose any set of POVMs acting on $\xx$. 
This relaxation makes the comparison much more subtle because different free-state inputs can be paired with different optimal measurements. 
For the resource theories of coherence and asymmetry, it was shown that the robustness still serves as a quantifier for the advantage in this setting~\cite{Napoli2016,Piani2016,Bu2017}. The proofs of these results rely on the simple structure of the two resources, allowing one to choose the set of subchannels in a way such that all free states remain invariant under the application of any subchannel, removing the need to explicitly maximize over all the measurement strategies. In fact, this can be used to establish a sufficient condition imposed at a more abstract level that allows this relation to hold in other resource theories; we formalize it as follows. Full proofs of the results in this section are provided in the Supplemental Material \cite{sm}.
\begin{pro}
Suppose $\rho\in D(\xx)$, and let $X=\sum_j x_j\dm{e_j}$ be the optimal witness in Eq.\eq{SDP_obj} for $\rho$.
If there exists a set of unitaries $\{U_i\}_{i=1}^d$ such that $\sum_i U_i\dm{e_j}U_i^{\dagger}=I,\forall j$ and $U_i\sigma U_i^{\dagger} = U_j\sigma U_j^{\dagger},\forall \sigma\in\ff(\xx),\forall i,j$, then 
\bann
 \max_{\{\Psi_i\}} \frac{\max_{\{M_i\}}p_{\succ}(\{\Psi_i\},\{M_i\},\rho)}{\max_{\sigma \in \ff(\xx),\{M_i\}} p_{\succ}(\{\Psi_i\},\{M_i\},\sigma)} = 1 + R_{\ff(\xx)}(\rho).
\label{eq:rob_gen_cond}
\eann
\label{prop:rob_gen_cond}
\end{pro}
One can easily verify that, for instance, coherence theory satisfies this condition, which recovers the result in \cite{Napoli2016,Piani2016,Bu2017}.

It could perhaps seem that one cannot expect the same relation to hold for theories with a more complex structure, as in general the measurement strategies could be chosen in a way which leads to better success probability with free states. However, rather surprisingly, it turns out that the robustness still acts as the exact quantifier of the operational advantage in this general setting in the resource theory of entanglement.
\begin{thm}
\label{thrm:rob_ent}
Let $\ff(\xx)={\rm SEP}(\xx_1\otimes\xx_2)$ where ${\rm SEP}(\xx_1\otimes\xx_2)$ is the set of separable states with respect to the bipartition between $\xx_1$ and $\xx_2$. 
Then, for any $\rho\in D(\xx_1\otimes\xx_2)$,
\bann
 \max_{\{\Psi_i\}} \frac{\max_{\{M_i\}}p_{\succ}(\{\Psi_i\},\{M_i\},\rho)}{\max_{\sigma \in \ff(\xx),\{M_i\}} p_{\succ}(\{\Psi_i\},\{M_i\},\sigma)} = 1 + R_{\ff(\xx)}(\rho).
\label{eq:rob_ent}
\eann
\end{thm}
  
One may then wonder if it is possible to extend this property to other resource theories. 
However, it appears that a possible generalization of \thrm{rob_ent} to other resources is rather nontrivial, even in the simplest cases such as single-qubit magic theory. The subtlety lies in upper bounding the denominator of the statement, which is maximized over all the possible input free states and measurements.
To remedy this, we consider a more restrictive, but still natural, situation where experimenters are free to choose independent measurement strategies but are constrained to use {\it free measurements} \cite{matthews_2010}.
We call a measurement constructed by the POVMs $\{M_i\}$ a free measurement if all the POVM elements are proportional to some free state, namely, $M_i\propto \sigma_i\ \forall i$ for $\sigma_i\in \ff(\xx)$. 
Under this restriction, we first find that the generalized robustness remains an exact quantifier for the resource theory of coherence.
\begin{pro}
\label{prop:rob_coh}
Let $\ff(\xx)=\mathcal{I}(\xx)$ where $\mathcal{I}(\xx)$ is the set of incoherent states with some preferred basis and $\mm_\ff$ be the set of free measurements with respect to $\ff(\xx)$. For any $\rho\in D(\xx)$,
\begin{equation}\begin{aligned}
 &\max_{\{\Psi_i\}} \frac{\max_{\{M_i\}\in \mm_\ff}p_{\succ}(\{\Psi_i\},\{M_i\},\rho)}{\max_{\sigma \in \ff(\xx),\{M_i\}\in \mm_\ff} p_{\succ}(\{\Psi_i\},\{M_i\},\sigma)} \nonumber\\ &= 1 + R_{\ff(\xx)}(\rho). \nonumber
\label{eq:rob_coh}
\end{aligned}\end{equation}
\end{pro}
If we further restrict the measurements to be rank-one, the same statement holds for single-qubit magic theory with pure input states.
\begin{pro}
\label{prop:rob_magic}
Let $\ff(\xx)={\rm STAB}(\xx)$ where ${\rm STAB}(\xx)$ is the set of stabilizer states defined on a single-qubit system and $\mm_\ff^1$ be the set of rank-one free measurements with respect to $\ff(\xx)$. For any pure state $\rho=\dm{\psi}\in D(\xx)$,
\begin{equation}\begin{aligned}
 &\max_{\{\Psi_i\}} \frac{\max_{\{M_i\}\in \mm_\ff^1}p_{\succ}(\{\Psi_i\},\{M_i\},\rho)}{\max_{\sigma \in \ff(\xx),\{M_i\}\in \mm_\ff^1} p_{\succ}(\{\Psi_i\},\{M_i\},\sigma)} \nonumber\\ &= 1 + R_{\ff(\xx)}(\rho). \nonumber
\label{eq:rob_magic}
\end{aligned}\end{equation}
\end{pro}
We note that an optimal task in \prop{rob_coh} is distinct from the phase discrimination game considered in Refs.\,\cite{Napoli2016,Piani2016,Bu2017}, which requires a non-free measurement. 
We show that the resourceful part in the measurement can be pushed into the subchannels so that the measurement becomes free.
This idea also works for the resource theory of magic in two-dimensional systems, but already the generalization beyond this case becomes much less straightforward.

{\textit{\textbf{Conclusions.}}} ---
We have shown that every resource state defined in any convex resource theory is useful in a channel discrimination task. 
It automatically gives an operational characterization to all resource states, including bound resources, in which the word ``resource'' gains an actual physical meaning. 
We have then found that the maximal advantage in the success probability of a class of subchannel discrimination problems is exactly quantified by the generalized robustness measure. 
Our result ensures that the generalized robustness measure always admits an operational interpretation in every convex resource theory.
We have finally considered relaxing the constraint on the allowed measurement:
for the case when the measurement strategies for the resource-state input and for any free-state input can be chosen independently, the generalized robustness still serves as the exact quantifier for the maximal advantage when the input states are entangled states; analogous results can be shown under the restriction of free measurements in the resource theories of coherence and single-qubit magic.

An important outstanding open question is : to what extent can the results of \thrm{rob_ent} and Propositions \ref{prop:rob_coh}--\ref{prop:rob_magic} be generalized, providing a more complete understanding of the generalized robustness as a quantifier of operational advantage in various subchannel discrimination tasks? Additionally, it would be interesting to establish a similar operational characterization of a resource measure related to $R_{\ff(\xx)}$ called the \textit{standard} robustness of a resource, where the optimization over $\tau \in D(\xx)$ is replaced with an optimization over $\tau \in \ff(\xx)$, and which is known to admit operational interpretations in the resource theories of entanglement \cite{brandao_2007,brandao_2011} and magic \cite{howard_2017}.

\begin{acknowledgments}
{\em Note added.} --- An analogous result to \thrm{rob_ent} has been independently obtained by Bae {\it et al.}~\cite{bae2018more}, where the authors considered specifically the case of local subchannels applied to a single party, and investigated the advantages which entanglement can provide in that setting.
Also, recently Skrzypczyk and Linden~\cite{skrzypczyk2018robustness} have conjectured a general picture relating robustness-based measures, discrimination tasks, and information-theoretic quantities, for which our results establish one of the connections. 

{\em Acknowledgments.} --- We are grateful to Marco Piani and Joonwoo Bae for fruitful discussions and sharing with us parts of their unpublished work \cite{bae2018more} related to \thrm{rob_ent}. R.T. acknowledges the support from NSF, ARO, IARPA, and the Takenaka Scholarship Foundation. B.R. and G.A. acknowledge financial support from the European Research Council (ERC) under the Starting Grant GQCOP (Grant No.~637352).
K. B. acknowledges the support of 
the Templeton Religion Trust under
grant TRT 0159 and Academic Awards for Outstanding Doctoral Candidates at Zhejiang University.
Z.-W. L. acknowledges support by  AFOSR, ARO, and Perimeter Institute for Theoretical Physics. Research at Perimeter Institute is supported by the Government of Canada through Industry Canada and by the Province of Ontario  through  the  Ministry of Research and Innovation.
\end{acknowledgments}

\bibliographystyle{apsrev4-1}
\bibliography{myref}

\newpage

\begin{appendix}
\begin{center}
{\bf Supplemental Material}
\end{center}
\section{Detecting noise with a non-Clifford gate}
We apply the result of Theorem \ref{thrm:rob_fixed} to the theory of magic and relate it to a problem of detecting a noise that comes with an implementation of a non-Clifford gate.
In the following, we discuss one specific example, but a similar argument can be applied to other situations as well. 

The result of Theorem \ref{thrm:rob_fixed} together with its proof tells that the standard magic state called $T$-state defined by $\ket{T}\coloneqq \frac{1}{\sqrt{2}}(\ket{0}+e^{i\pi/4}\ket{1})$ gives the maximal advantage for discriminating the noiseless channel $\Lambda_0(\cdot)=I \cdot I$ and the phase flip channel $\Lambda_1(\cdot)=Z\cdot Z$, where $Z$ is the Pauli-$Z$ operator, when one is to use the projective measurement defined by $M_0=\dm{T}$, $M_1=\dm{\bar{T}}$ where $\ket{\bar{T}}=Z\ket{T}$.  
Let $U_{NC}\coloneqq \exp(-i\frac{\pi}{4}\frac{X-Y}{\sqrt{2}})$, which is the $\pi/2$ rotation with respect to the axis $\frac{1}{\sqrt{2}}(1,-1,0)$ on the Bloch sphere. 
It is a non-Clifford unitary, and it realizes the universal quantum computation together with the Clifford gates. 
Since $\ket{T}=U_{NC}^{\dagger}\ket{0}$ and $\ket{\bar{T}}=U_{NC}^{\dagger}\ket{1}$, measuring with POVMs $M_0$ and $M_1$ is equivalent to the computational basis measurement following the application of $U_{NC}$.
Thus, Theorem \ref{thrm:rob_fixed} implies that $T$-state is useful to detect the phase flip error prior to the non-Clifford unitary $U_{NC}$ when one is restricted to the computational basis measurement. 

It would be of practical relevance since verifying an error-free implementation of a non-Clifford gate is arguably important for fault-tolerant quantum computation, and the phase flip error, which is a source of decoherence, is a common type of error for many architectures.   
    
\section{Proofs of the results}

The proofs of some results in the manuscript rely on the following characterization of the optimality conditions for measurements in state discrimination.
\begin{lem}[\cite{HOLEVO1973337,Yuen1975OptimumTO,Helstrom}]
For the minimum-error state discrimination for the ensemble $\{q_i,\rho_i\}$ where one is to maximize $p_{\succ}=\sum_i q_i\Tr(M_i\rho_i)$, a set of POVMs $\{M_i\}$ is optimal if and only if
\ba
 \sum_iq_i\rho_iM_i-q_j\rho_j\geq 0,\ \forall j.
\label{eq:state_cond2}
\ea
\end{lem}
\vspace*{.5\baselineskip}

\subsection{Proof of Proposition \ref{prop:rob_gen_cond}}

\let\temp\thethm
\renewcommand{\thethm}{\ref{prop:rob_gen_cond}}

\begin{pro}
Suppose $\rho\in D(\xx)$, and let $X=\sum_j x_j\dm{e_j}$ be the optimal witness in \eq{SDP_obj} for $\rho$.
If there exists a set of unitaries $\{U_i\}_{i=1}^d$ such that $\sum_i U_i\dm{e_j}U_i^{\dagger}=I,\forall j$ and $U_i\sigma U_i^{\dagger} = U_j\sigma U_j^{\dagger},\forall \sigma\in\ff(\xx),\forall i,j$, then 
\bann
 \max_{\{\Psi_i\}} \frac{\max_{\{M_i\}}p_{\succ}(\{\Psi_i\},\{M_i\},\rho)}{\max_{\sigma \in \ff(\xx),\{M_i\}} p_{\succ}(\{\Psi_i\},\{M_i\},\sigma)} = 1 + R_{\ff(\xx)}(\rho).
\eann
\end{pro}

\let\thethm\temp
\addtocounter{thm}{-1}

\begin{proof}
Since it can be easily seen that the left-hand side is less than or equal to the right-hand side, it suffices to show the converse holds. 
 Let $d={\rm dim}\xx$ and take $\Psi_i(\cdot)=\frac{1}{d}U_i\cdot U_i^{\dagger}$ and $M_i=\frac{1}{\Tr(X)}U_iXU_i^{\dagger}$. 
As we saw in the proof of \thrm{rob_fixed}, it gives $p_{\succ}(\{\Psi_i\},\{M_i\},\rho) = \Tr(\rho X)/\Tr(X)$. 
We shall see that under the assumption of the statement this measurement (in fact, any measurement) is optimal for any free-state inputs. 
Let $\sigma_i\equiv U_i\sigma U_i^{\dagger}$ and consider the state discrimination for the ensemble $\{1/d,\sigma_i\}_{i=1}^d$. The assumption $U_i\sigma U_i^{\dagger}=U_j\sigma U_j^{\dagger}$ implies $\sigma_i = \sigma_j, \forall i,j$, so \eq{state_cond2} is satisfied as
\ba
 \frac{1}{d}\sum_i \sigma_i M_i - \frac{1}{d}\sigma_j = \frac{1}{d}\sigma_j\left(\sum_i M_i - I\right)=0.
\ea
This measurement gives the success probability $\Tr(\sigma X)/\Tr(X)\leq 1/\Tr(X)$, which leads to the statement of the Proposition.
\end{proof}

\subsection{Proof of Theorem \ref{thrm:rob_ent}}

\let\temp\thethm
\renewcommand{\thethm}{\ref{thrm:rob_ent}}

\begin{thm}
Let $\ff(\xx)={\rm SEP}(\xx_1\otimes\xx_2)$ where ${\rm SEP}(\xx_1\otimes\xx_2)$ is the set of separable states with respect to the bipartition between $\xx_1$ and $\xx_2$. 
Then, for any $\rho\in D(\xx_1\otimes\xx_2)$,
\bann
 \max_{\{\Psi_i\}} \frac{\max_{\{M_i\}}p_{\succ}(\{\Psi_i\},\{M_i\},\rho)}{\max_{\sigma \in \ff(\xx),\{M_i\}} p_{\succ}(\{\Psi_i\},\{M_i\},\sigma)} = 1 + R_{\ff(\xx)}(\rho).
\eann
\end{thm}

\let\thethm\temp
\addtocounter{thm}{-1}

\begin{proof}
In the same way as in the proof of \thrm{rob_fixed}, it can be shown that the left-hand side is less than or equal to the right-hand side, so it suffices to show that the left-hand side is greater than or equal to the right-hand side. 
We assume ${\rm dim}\xx_1\geq{\rm dim}\xx_2=d$ without loss of generality. Separability-preserving (non-entangling) channels are the free operations in the resource theory of entanglement defined as all CPTP maps such that $\sigma \in {\rm SEP}(\xx_1\otimes\xx_2) \Rightarrow \Lambda(\sigma) \in {\rm SEP}(\xx_1\otimes\xx_2)$. 
For any $\rho$, there exists a separability-preserving channel $\Lambda$ such that $d\Tr[\Lambda(\rho)\dm{\Phi^+_d}]=1+R_{\ff(\xx)}(\rho)$ where $\ket{\Phi^+_d}=\frac{1}{\sqrt{d}}\ket{ii}$ \cite{regula_2018-1}.
Let us then define $U_i=I\otimes P_i$ for $i=1,\dots,d^2$ where $P_i$ is the $i$th Pauli operator with respect to basis $\{\ket{i}\}$, and the subchannel $\Psi_i(\cdot) = \frac{1}{d^2}U_i\Lambda(\cdot) U_i^{\dagger}$.
This choice of $\{U_i\}$ satisfies $\sum_iU_i\dm{\Phi^+_d}U_i^{\dagger}=I$ because local Pauli operators map one Bell basis to another Bell basis.
Thus, we can take $M_i=U_i\dm{\Phi^+_d}U_i^{\dagger}$ as a valid POVM, and it realizes that $p_{\succ}(\{\Psi_i\},\{M_i\},\rho)=\Tr\left[\Lambda(\rho)\dm{\Phi^+_d}\right]=(1+R_{\ff(\xx)}(\rho))/d$.

Now, our goal is to show that for this choice of $\{\Psi_i\}$, $\max_{\sigma \in \ff(\xx),\{M_i\}}p_{\succ}(\{\Psi_i\},\{M_i\},\sigma)\leq 1/d$. We obtain
\begin{equation*}
\begin{aligned}
 &\max_{\sigma \in \ff(\xx),\{M_i\}}p_{\succ}(\{\Psi_i\},{M_i},\sigma) \\
 &\quad= \max_{\sigma \in \ff(\xx),\{M_i\}} \sum_{i=1}^{d^2} \Tr_{\xx_1\xx_2}[M_i\Psi_i(\sigma)]\\
 &\quad= \max_{\sigma \in \ff(\xx),\{M_i\}} \sum_{i=1}^{d^2} \frac{1}{d^2}\Tr_{\xx_1\xx_2}[M_iU_i\Lambda(\sigma)U_i^{\dagger}]\\
 &\quad\leq \max_{\tilde{\sigma} \in \ff(\xx),\{M_i\}} \sum_{i=1}^{d^2} \frac{1}{d^2}\Tr_{\xx_1\xx_2}\left[M_iU_i\tilde{\sigma} U_i^{\dagger}\right]\\
  &\quad= \max_{\ket{\phi_{\xx_1}},\ket{\phi_{\xx_2}}}\max_{\{M_i\}} \sum_{i=1}^{d^2} \frac{1}{d^2}\times \\
 &\quad\quad\Tr_{\xx_1\xx_2}\left[M_iU_i\left(\dm{\phi_{\xx_1} }\otimes\dm{\phi_{\xx_2}}\right) U_i^{\dagger}\right]\\
 &\quad\leq \max_{\ket{\phi_{\xx_2}}} \max_{\{N_i\}} \sum_{i=1}^{d^2} \frac{1}{d^2}\Tr_{\xx_2}\left[N_iP_i\dm{\phi_{\xx_2}}P_i\right].
\end{aligned}
\end{equation*}
In the first inequality, we used that $\Lambda$ is separability-preserving, and thus the set of output states of $\Lambda$ with separable-state inputs is contained in the set of separable states. 
In the third equality, we used that the maximum of a linear functional over the separable states always occurs at an extreme point (a pure product state).
To get the last inequality, note that $N_i=\bra{\phi_{\xx_1}}M_i\ket{\phi_{\xx_1}}$ forms a valid set of POVMs acting on $\xx_2$ because clearly $N_i\geq 0$, and $\sum_i N_i = \sum_i \bra{\phi_{\xx_1}}M_i\ket{\phi_{\xx_1}}=\bra{\phi_{\xx_1}}I_{\xx_1\xx_2}\ket{\phi_{\xx_1}}=I_{\xx_2}$. 
The inequality then follows because the set of measurements with this form of POVMs is a subset of all the valid POVM measurements acting on $\xx_2$.
We will show that the quantity in the last line equals to $1/d$ by constructing a specific measurement strategy that achieves it and show that that strategy is optimal. 
Consider $N_i=\frac{1}{d}P_i\dm{\phi_{\xx_2}}P_i$. 
It forms a valid POVMs, i.e. $\sum_i N_i=I$, due to the Pauli twirling property.
Then, it is easily seen that $\sum_{i=1}^{d^2} \frac{1}{d^2}\Tr_{\xx_2}\left[N_iP_i\dm{\phi_{\xx_2}}P_i\right]=1/d$.
Now, we shall see that for any given $\ket{\phi_{\xx_2}}$, this choice of $\{N_i\}$ is optimal. 
Note that once $\ket{\phi_{\xx_2}}$ is given, the problem is reduced to the state discrimination for the ensemble $\{1/d^2,\ket{\phi_{\xx_2}^i}\}$ where $\ket{\phi_{\xx_2}^i}\equiv P_i\ket{\phi_{\xx_2}}$. 

We can check that our choice of $\{N_i\}$ satisfies the condition \eq{state_cond2} as follows. Note that $N_i=\frac{1}{d}\dm{\phi_{\xx_2}^i}$. Then, we get
\bann
&&\sum_i \frac{1}{d^2}\dm{\phi_{\xx_2}^i}N_i-\frac{1}{d^2}\dm{\phi_{\xx_2}^j}\\&=&\sum_i \frac{1}{d^3} \dm{\phi_{\xx_2}^i}-\frac{1}{d^2}\dm{\phi_{\xx_2}^j}\\
&=&\sum_i \frac{1}{d^3} P_i\dm{\phi_{\xx_2}}P_i-\frac{1}{d^2}\dm{\phi_{\xx_2}^j}\\
&=& \frac{1}{d^2}\left(I - \dm{\phi_{\xx_2}^j}\right) \geq 0
\eann
We just showed that for this specific choice of $\{\Psi_i\}$, $\max_{\{M_i\}}p_{\succ}(\{\Psi_i\},{M_i},\ket{\psi})= (1+R_{\ff(\xx)})/d$,  and $\max_{\sigma \in \ff(\xx),\{M_i\}}p_{\succ}(\{\Psi_i\},{M_i},\sigma)= 1/d$, which concludes the proof.
\end{proof}

\subsection{Proof of Proposition \ref{prop:rob_coh}}

\let\temp\thethm
\renewcommand{\thethm}{\ref{prop:rob_coh}}

\begin{pro}
Let $\ff(\xx)=\mathcal{I}(\xx)$ where $\mathcal{I}(\xx)$ is the set of incoherent states with some preferred basis and $\mm_\ff$ be the set of free measurements with respect to $\ff(\xx)$. For any $\rho\in D(\xx)$,
\begin{equation}\begin{aligned}
 &\max_{\{\Psi_i\}} \frac{\max_{\{M_i\}\in \mm_\ff}p_{\succ}(\{\Psi_i\},\{M_i\},\rho)}{\max_{\sigma \in \ff(\xx),\{M_i\}\in \mm_\ff} p_{\succ}(\{\Psi_i\},\{M_i\},\sigma)} \nonumber\\ &= 1 + R_{\ff(\xx)}(\rho). \nonumber
\label{eq:rob_coh}
\end{aligned}\end{equation}
\end{pro}

\let\thethm\temp
\addtocounter{thm}{-1}

\begin{proof}
Again, it can be easily seen that the left-hand side is less than or equal to the right-hand side, so we will show that the converse holds. 
Suppose $\dim\xx=d$. It was shown in \cite{Bu2017,regula_2017} that for any $\rho\in D(\xx)$, there exists a coherence non-generating channel $\Lambda$ such that $1+R_{\ff(\xx)}=d\Tr[\Lambda(\rho)\dm{w}]$ where $\ket{w}$ takes the form $\ket{w}=\frac{1}{\sqrt{d}}\sum_{j=0}^{d-1} \ket{j}$.
Let $H_d=\frac{1}{\sqrt{d}}\sum_{k,j=0}^{d-1} \zeta^{kj}\ketbra{k}{j}$ where $\zeta = e^{\frac{2\pi i}{d}}$. 
Take $U_0=H_d$ and $U_i=X^iU_0$ for $1\leq i \leq d-1$ where $X=\sum_j \ketbra{j+1\ ({\rm mod}\ d)}{j}$. 
Define $\Psi_i(\cdot)=\frac{1}{d} U_i \Lambda(\cdot) U_i^{\dagger}$ and $M_i=U_i\dm{w}U_i^{\dagger}=\dm{i}$, then $p_{\succ}(\{\Psi_i\},\{M_i\},\rho)=(1+R_{\ff(\xx)}(\rho))/d$, and clearly $\{M_i\}\in \mm_\ff$.
Therefore, it suffices to show $\max_{\sigma \in \ff(\xx),\{M_i\}\in \mm_\ff} p_{\succ}(\{\Psi_i\},\{M_i\},\sigma)\leq 1/d$. 
In a similar way we took in the proof of \thrm{rob_ent}, we obtain
\begin{equation}
 \begin{aligned}
 &\max_{\sigma \in \ff(\xx),\{M_i\}\in \mm_\ff}p_{\succ}(\{\Psi_i\},\{M_i\},\sigma) \\
 &\quad= \max_{\sigma \in \ff(\xx),\{M_i\}\in \mm_\ff} \sum_{i=1}^{d} \Tr[M_i\Psi_i(\sigma)]\\
 &\quad= \max_{\sigma \in \ff(\xx),\{M_i\}\in \mm_\ff} \sum_{i=1}^{d} \frac{1}{d}\Tr[M_iU_i\Lambda(\sigma)U_i^{\dagger}]\\
 &\quad\leq \max_{\tilde{\sigma} \in \ff(\xx),\{M_i\}\in \mm_\ff} \sum_{i=1}^{d} \frac{1}{d}\Tr\left[M_iU_i\tilde{\sigma} U_i^{\dagger}\right]\\
  &\quad= \max_{\ket{l}}\max_{\{M_i\}\in \mm_\ff} \sum_{i=1}^{d} \frac{1}{d}\Tr\left[M_iU_i\left(\dm{l}\right) U_i^{\dagger}\right]
\end{aligned}
\label{eq:coh_proof}
\end{equation}
In the first inequality, we used that $\Lambda$ is coherence-nongenerating, and thus the set of output states of $\Lambda$ with incoherent inputs is contained in the set of incoherent states. 
In the third equality, we used that the maximum over the incoherent states always occurs at a pure state. 
Since $\{M_i\}\in \mm_\ff$, $M_i$ takes the form $M_i=\sum_j c_{i,j}\dm{j}$ where $c_{i,j}\geq 0,\,\sum_i c_{i,j}=1\ \forall j$. It gives $\sum_{i=1}^d\Tr\left[M_iU_i\left(\dm{l}\right) U_i^{\dagger}\right]=\sum_{i,j}c_{i,j}|\braket{j|U_i|l}|^2$.
Since each $U_i$ acts on pure incoherent states as $U_i\ket{l}=\frac{1}{\sqrt{d}}\sum_{k=0}^{d-1} \zeta^{lk}\ket{k+i\ ({\rm mod}\ d)}$, we get $|\braket{j|U_i|l}|^2=1/d\ \forall i,j,l$.
Thus, the last equality of \eq{coh_proof} becomes
\bann
 \max_{\ket{l}}\max_{\{M_i\}\in \mm_\ff} \sum_{i=1}^{d} \frac{1}{d}\Tr\left[M_iU_i\left(\dm{l}\right) U_i^{\dagger}\right] = \frac{1}{d^2}\sum_{i,j} c_{i,j} = \frac{1}{d},
\eann
which is what we wanted to show.  
\end{proof}

\subsection{Proof of Proposition \ref{prop:rob_magic}}

\let\temp\thethm
\renewcommand{\thethm}{\ref{prop:rob_magic}}

\begin{pro}
Let $\ff(\xx)={\rm STAB}(\xx)$ where ${\rm STAB}(\xx)$ is the set of stabilizer states defined on a single-qubit system and $\mm_\ff^1$ be the set of rank-one free measurements with respect to $\ff(\xx)$. For any pure state $\rho=\dm{\psi}\in D(\xx)$,
\begin{equation}\begin{aligned}
 &\max_{\{\Psi_i\}} \frac{\max_{\{M_i\}\in \mm_\ff^1}p_{\succ}(\{\Psi_i\},\{M_i\},\rho)}{\max_{\sigma \in \ff(\xx),\{M_i\}\in \mm_\ff^1} p_{\succ}(\{\Psi_i\},\{M_i\},\sigma)} \nonumber\\ &= 1 + R_{\ff(\xx)}(\rho). \nonumber
\label{eq:rob_magic}
\end{aligned}\end{equation}
\end{pro}

\let\thethm\temp
\addtocounter{thm}{-1}

\begin{proof}
Once again, it can be easily seen that the left-hand side is less than or equal to the right-hand side, so we will show that the converse holds.
 Let $X=c\dm{w},\ket{w}\in D(\xx)$ be an optimal witness satisfying \eq{SDP_cond1} and \eq{SDP_cond2}. 
 Note that the optimal witness for pure state can be always taken in this rank-one form~\cite{regula_2018}.
 Because of \eq{SDP_cond2} and optimality, $c^{-1}=\max_{\ket{\phi}\in \ff(\xx)}|\braket{\phi|w}|^2$. Let $\ket{\tilde{\phi}}\in\ff(\xx)$ be a state that achieves this maximum.
 Because of the symmetry of the stabilizer hull, it suffices to prove the statement only for the situation where all the elements of the Bloch coordinate of $\ket{w}$ are positive and $\ket{\tilde{\phi}}=\ket{0}$.
Let us write the Bloch coordinate of $\ket{w}$ as $(\sin\theta \cos\varphi, \sin\theta \sin\varphi, \cos\theta)$. Then the above restriction limits the domain of $\theta$ and $\varphi$ as 
\ba
 0\leq \theta \leq \arccos(1/\sqrt{3}),\ 0\leq \varphi \leq \pi/2.
\label{eq:angle_domain}
\ea
Note that the upper bound of $\theta$ depends on $\varphi$; $\theta$ is restricted in the way that the closest pure stabilizer state is $\ket{0}$. $\theta = \arccos(1/\sqrt{3})$ can be only achieved when $\varphi=\pi/4$. 
 Set ${\bs n}=(\sin\varphi,-\cos\varphi,0)$ and take $U_1\equiv R_{\bs n}(\theta)$ where $R_{\bs n}(\theta)=\exp(-i(\theta/2) {\bs n}\cdot{\bs \sigma})$ denotes the single-qubit rotation by $\theta$ with respect to axis ${\bs n}$.
Also, define $U_2=R_{\bs n}(\theta + \pi)$.
By definition, it realizes $U_1\ket{w}=\ket{0}$ and $U_2\ket{w}=\ket{1}$, so they satisfy $\sum_i U_i\dm{w}U_i^{\dagger}=I$.  
Thus, by taking $\Psi_i(\cdot)$ = $\frac{1}{2}U_i\cdot U_i^{\dagger}$ and $M_i=\frac{1}{c}U_iXU_i^{\dagger}$, we obtain $p_{\succ}(\{\Psi_i\},\{M_i\},\ket{\psi})=\left[1+R_{\ff(\xx)}(\ket{\psi})\right]/c=1$.
Note that $\{M_i\}\in \mm_\ff^1$ because $M_0=\dm{0}$ and $M_1=\dm{1}$.

Now, it suffices to show that for this choice of $\{\Psi_i\}$, $\max_{\sigma \in \ff(\xx),\{M_i\}\in \mm_\ff^1} p_{\succ}(\{\Psi_i\},\{M_i\},\sigma)\leq 1/c=|\braket{0|w}|^2$.
This can be significantly simplified as follows. First, due to the form of $p_{\succ}$, maximum only occurs at pure free states. Moreover, thanks to the symmetry of the setting we are considering, it suffices to only consider $\sigma = \dm{0},\dm{+}$. 
Also, since $\sum_i M_i = I$, the maximum occurs at two-value projective measurement along either $X,Y,Z$ direction.
At the end, what we need to show is just that for $\ket{\phi}\in \{\ket{0},\ket{+}\}$,
\bann
 \max_{\substack{\ket{\xi_1},\ket{\xi_2}\in {\rm STAB}\\ \braket{\xi_2|\xi_1}=0}} \frac{1}{2}\sum_{i=1}^2\Tr[\dm{\xi_i}U_i\dm{\phi}U_i^{\dagger}] \leq |\braket{0|w}|^2.
\eann
However, for $\ket{\phi}=\ket{0}$, it is almost trivial (if you draw a picture) that the left-hand side equals to the right-hand side where the maximum happens when $\ket{\xi_0}=\ket{0}$ and $\ket{\xi_1}=\ket{1}$. 
For $\ket{\phi}=\ket{+}$, it can be checked that the above inequality holds for three choices of the measurement by the following straightforward calculations.
It is convenient to recall that $U_1=R_{\bs n}(\theta)$, $U_2=R_{\bs n}(\theta + \pi)$, and
\bann
\begin{cases}
R_{\bs n}(\theta)\ket{0}=\cos(\theta/2)\ket{0}-e^{i\varphi}\sin(\theta/2)\ket{1} \\
R_{\bs n}(\theta)\ket{1}=e^{-i\varphi}\sin(\theta/2)\ket{0}+\cos(\theta/2)\ket{1}
\end{cases}
\eann
\bann
\begin{cases}
R_{\bs n}(\theta+\pi)\ket{0}=-\sin(\theta/2)\ket{0}-e^{i\varphi}\cos(\theta/2)\ket{1}\\
R_{\bs n}(\theta+\pi)\ket{1}=e^{-i\varphi}\cos(\theta/2)\ket{0}-\sin(\theta/2)\ket{1}
\end{cases}
\eann

\begin{itemize}
\item For $\ket{\xi_1}=\ket{0},\ket{\xi_2}=\ket{1}$ (clearly $\ket{\xi_1}=\ket{1},\ket{\xi_2}=\ket{0}$ does not achieve the maximum.)
\bann
 |\braket{0|R_{\bs n}(\theta)|+}|^2&=&\frac{1}{2}|\cos(\theta/2)+e^{-i\varphi}\sin(\theta/2)|^2 \\&=& \frac{1}{2}(1+\cos\varphi\sin\theta),
\eann 
and
\bann
 |\braket{1|R_{\bs n}(\theta+\pi)|+}|^2&=&\frac{1}{2}|-e^{i\varphi}\cos(\theta/2)-\sin(\theta/2)|^2 \\&=& \frac{1}{2}(1+\cos\varphi\sin\theta).
\eann
Thus,
\bann
 &&\frac{1}{2}|\braket{0|R_{\bs n}(\theta)|+}|^2 + \frac{1}{2}|\braket{1|R_{\bs n}(\theta+\pi)|+}|^2 \\&=& \frac{1}{2}(1+\cos\varphi\sin\theta) 
 \\&\leq& \frac{1}{2}(1+\cos\theta)=|\braket{0|w}|^2
\eann
The inequality is due to the assumption that $|\braket{0|w}|^2\geq |\braket{+|w}|^2$.

\item $\ket{\xi_1}=\ket{+},\ket{\xi_2}=\ket{-}$ (clearly $\ket{\xi_1}=\ket{-},\ket{\xi_2}=\ket{+}$ does not achieve the maximum.)
\bann
 &&|\braket{+|R_{\bs n}(\theta)|+}|^2\\
 &=&\frac{1}{4}|\cos(\theta/2)-e^{i\varphi}\sin(\theta/2)+e^{-i\varphi}\sin(\theta/2)+\cos(\theta/2)|^2 \\
 &=& |\cos(\theta/2)-i\sin\varphi\sin(\theta/2)|^2\\
&=&\cos^2(\theta/2)+\sin^2\varphi\sin^2(\theta/2)\\
&=& 1-\cos^2\varphi\sin^2(\theta/2),  
\eann 
and
\bann
 &&|\braket{-|R_{\bs n}(\theta+\pi)|+}|^2\\&=&\frac{1}{4}|-\sin(\theta/2)+e^{-i\varphi}\cos(\theta/2)+e^{i\varphi}\cos(\theta/2)+\sin(\theta/2)|^2 \\
 &=& \cos^2\varphi\cos^2(\theta/2)
\eann
Thus,
\bann
 &&\frac{1}{2}|\braket{+|R_{\bs n}(\theta)|+}|^2+\frac{1}{2}|\braket{-|R_{\bs n}(\theta+\pi)|+}|^2 \\&=& \frac{1}{2}(1+\cos^2\varphi\cos\theta) \\
 &\leq& \frac{1}{2}(1+\cos\theta)=|\braket{0|w}|^2.
\eann

\item $\ket{\xi_1}=\ket{+y},\ket{\xi_2}=\ket{-y}$ (clearly $\ket{\xi_1}=\ket{-y},\ket{\xi_2}=\ket{+y}$ does not achieve the maximum.)
\bann
 &&|\braket{+y|R_{\bs n}(\theta)|+}|^2\\&=&\frac{1}{4}|\cos(\theta/2)+e^{-i\varphi}\sin(\theta/2)+ie^{i\varphi}\sin(\theta/2)-i\cos(\theta/2)|^2 \\
 &=& \frac{1}{4}|(1-i)\cos(\theta/2)+(e^{-i\varphi}+ie^{i\varphi})\sin(\theta/2)|^2\\
&=& \frac{1}{4}|\sqrt{2}\cos(\theta/2)-2i\sin(\varphi-\pi/4)\sin(\theta/2)|^2\\
&=&\frac{1}{2}\left[\cos^2(\theta/2)+2\sin^2(\varphi-\pi/4)\sin^2(\theta/2)\right]\\
&=&\frac{1}{2}\left[\cos^2(\theta/2)+2\cos^2(\varphi+\pi/4)\sin^2(\theta/2)\right],
\eann
and
\bann
&& |\braket{-y|R_{\bs n}(\theta+\pi)|+}|^2\\&=&\frac{1}{4}|-\sin(\theta/2)+e^{-i\varphi}\cos(\theta/2)-ie^{i\varphi}\cos(\theta/2)-i\sin(\theta/2)|^2 \\
 &=& \frac{1}{4}|-(1+i)\sin(\theta/2)+(e^{-i\varphi}-ie^{i\varphi})\cos(\theta/2)|^2\\
 &=& \frac{1}{4}|\sqrt{2}\sin(\theta/2)+2i\sin(\varphi+\pi/4)\cos(\theta/2)|^2\\
&=&\frac{1}{2}\left[\sin^2(\theta/2)+2\sin^2(\varphi+\pi/4)\cos^2(\theta/2)\right]
\eann 
Thus,
\bann
 &&\frac{1}{2}|\braket{+y|R_{\bs n}(\theta)|+}|^2+\frac{1}{2}|\braket{-y|R_{\bs n}(\theta+\pi)|+}|^2 \\
 &=& \frac{1}{4}\left[1+2\cos^2(\theta/2)-2\cos^2(\varphi+\pi/4)\cos\theta\right]\\
  &=& \frac{1}{4}\left[2-\cos(2\varphi+\pi/2)\cos\theta\right]\\
    &=& \frac{1}{2}\left[1+\frac{1}{2}\sin(2\varphi)\cos\theta\right]\\
    &\leq& \frac{1}{2}(1+\cos\theta)=|\braket{0|w}|^2.
\eann
\end{itemize}
\end{proof}
\end{appendix}
\end{document}